\documentclass[preprint,12pt,authoryear]{elsarticle}

\usepackage{graphicx}
\usepackage{amssymb}
\usepackage{amsmath}
\usepackage{amsthm}
\newtheorem{theorem}{Theorem}[section]

\begin{document}

\begin{frontmatter}

\title{Strong uniform consistency of nonparametric estimation for quantile-based entropy function under length-biased sampling} 

\author[stat]{Vaishnavi Pavithradas}
\author[stat]{Rajesh G.\corref{cor1}}
\cortext[cor1]{Corresponding author}
\ead{rajeshgstat@gmail.com}
\affiliation[stat]{organization={Department of Statistics},
            addressline={Cochin University of Science and Technology}, 
            city={Cochin},
            postcode={682022}, 
            state={Kerala},
            country={India}}

\begin{abstract}
For studies in reliability, biometry, and survival analysis, the length-biased distribution is often well-suited for certain natural sampling plans. In this paper, we study the strong uniform consistency of two nonparametric estimators for the quantile-based Shannon entropy in the context of length-biased data. A simulation study is conducted to examine the behavior of the estimators in finite samples, followed by a comparative analysis with existing estimators. Furthermore, the usefulness of the proposed estimators is evaluated using a real dataset.

\end{abstract}


\end{frontmatter}


\section{Introduction}
	\noindent
The entropy of a random variable is a function that plays a crucial role in quantifying the randomness, unpredictability, or uncertainty associated with probability distributions. Since Shannon's groundbreaking work in 1948, Shannon entropy has been defined as the average amount of information required to represent an event of a random variable from its probability distribution. Originally rooted in thermodynamics, entropy has gained significant relevance across various disciplines, including biology, chemistry, information theory, synergetics, and many other scientific fields. If $X$ is a non-negative random variable, with a continuous pdf $f(\cdot)$ and cdf $F(\cdot)$, then the Shannon differential entropy is defined by, as given in \cite{shannon1948mathematical} :
\begin{equation}\label{Sentropy}
    H(X) = -E(\log f(X)) = - \int (\log f(x))f(x)\,dx
\end{equation}
The fundamental properties of differential entropy are discussed in Chapter 9 of \cite{cover1999elements}. Various researchers have proposed methods for estimating \eqref{Sentropy}; see \cite{hall1993estimation}, \cite{ahmad1976nonparametric}, \cite{vasicek1976test}, and the references therein for the details. For a comprehensive review of nonparametric approaches to estimating Shannon entropy, refer to \cite{beirlant1997nonparametric}.\\
The measure \eqref{Sentropy} is well-defined, and their applications are primarily based on either the probability density function (pdf) or the survival function. However, in statistical modeling and data analysis, an equivalent and alternative approach involves the use of the quantile function (QF). A quantile is a value that corresponds to a specified proportion of a sample or population \cite{gilchrist2000statistical}. Mathematically, it is expressed as:
\begin{equation}\label{quantilefunctn}
    Q(u)=inf\left\{ x : F(x) \geq u \right\},~~0 \leq u \leq 1 
\end{equation}
When $F$ is continuous from \eqref{quantilefunctn} we have $FQ(u)=u$, which is a composite function $F(Q(u))$. The significance of quantile functions (QFs) and concepts derived from them is well recognized in exploratory data analysis and other applied statistics areas. QFs are often preferred because they are less affected by extreme observations, making them particularly useful for analysis even with a limited amount of information. For comprehensive studies on the properties and applications of QFs in modeling and analyzing lifetime data, refer to \cite{nair2008total}, \cite{nair2009quantile}, \cite{sankaran2009nonparametric}, \cite{nair2012modelling}, \cite{sankaran2010quantile}, and the references cited therein. \cite{sunoj2012quantile} introduced a quantile-based version of Shannon entropy, which is defined as:
\begin{equation}\label{qse}
    \xi(X)= \int_{0}^{1} \log (q(u))\, du,
\end{equation}
where $q(u)$ denotes the quantile density function. 
Firstly, in contrast to \eqref{Sentropy}, the computation of \eqref{qse} is relatively straightforward, particularly in cases where the distribution functions are complex or lack a tractable form, while their quantile functions are simpler. Secondly, it presents an alternative method for studying entropy functions. Quantile functions have certain features that probability distributions do not, making them useful for measuring uncertainty using \eqref{qse}. Unlike the entropy measure in \eqref{Sentropy}, which depends on residual and past lifetime variables, the quantile-based measure gives an explicit formula that uniquely identifies the underlying distribution, as shown in \cite{sunoj2012quantile}. These observations have led many researchers to explore the quantile-based approach to various entropy functions, yielding new methodologies, results, and properties.\\
Length-biased data frequently arise in observational studies when the sampled observations are not randomly chosen from the target population but are selected with a probability proportional to their length ( \cite{simon1980length}; \cite{vardi1982nonparametric}; \cite{vardi1985empirical}; \cite{zelen2004forward}). Such data frequently arise in contexts like renewal processes, epidemiological cohort studies, and screening programs in chronic disease research (\cite{zelen1969theory}).  The probability density function of a length-biased random variable $g$, is related to the underlying density $f$, say, by,
    \begin{equation}\label{lbpdf}
    g_{Y} (x) = \frac{x f(x)}{\mu}, \; x > 0,
\end{equation}
 Here we assume that $\mu=\int xf(x)\;dx$ is finite. Several studies on nonparametric estimation under length-biased sampling have been conducted by authors such as \cite{vardi1982nonparametric}, \cite{horvath1985estimation}, and \cite{jones1991kernel}.

\cite{akbari2019nonparametric} established the strong uniform consistency of nonparametric estimators for the quantile density function under length-biased sampling. \cite{Oliazadeh2020ANO} proposed a nonparametric estimator for Shannon entropy and demonstrated its strong consistency in the context of length-biased data. Nonparametric kernel estimators for the Shannon differential entropy function under length-biased sampling were introduced by \cite{Rajesh2021KernelEO}, who also derived their asymptotic properties. It is important to note that the length-biased sampling model can be seen as a special case of the left-truncated and right-censored (LTRC) model, particularly when the left truncation follows a uniform distribution and there is no right censoring. However, the three nonparametric estimators for the quantile-based Shannon differential entropy function proposed by \cite{subhash2023nonparametric} are not suitable for length-biased data. Motivated by this gap, the main objective of this paper is to propose two estimators for the quantile-based entropy function tailored to length-biased observations. We also establish the strong uniform convergence of our estimators, inspired by the proof in \cite{Bouzebda2008UniforminbandwidthCF}.

The structure of the paper is as follows: Section \ref{sec2} presents the quantile density estimation method for length-biased distributions proposed by \cite{akbari2019nonparametric}. Section \ref{sec3} presents two nonparametric estimators based on these quantile density functions mentioned in Sec. \ref{sec2} and establishes their strong uniform convergence. Section \ref{sec4}, dedicated to numerical analysis, is divided into three parts: Subsection \ref{simulation} examines the behavior of the proposed estimator through a simulation study, Subsection \ref{compstudy} compares our proposed estimators with existing estimators, and Subsection \ref{realdata} validates their performance using a real dataset. Finally, Section \ref{conclusion} summarizes the key findings of this study.
 
 \section{Quantile density estimation}\label{sec2}
 \noindent
 This section begins with an overview of preliminary concepts related to estimating functions in the length-biased context.\\
 Let $X$ be a non-negative random variable with a probability density function $f(\cdot)$, with an unknown finite mean $\mu$. Let $Y$ denote the corresponding length-biased random variable derived from $X$. Consider a random sample $\left(Y_{1}, Y_{2}, \ldots, Y_{n}\right)$ consisting of
 independent and identically distributed random variable from $Y$. Suppose $g$ is the corresponding density function of $Y$. Based on this, \cite{jones1991kernel} proposed the following kernel density estimator for the original density $f$ given by,
 \begin{equation}
    \label{jones}
		f_n(x)=\frac{\mu_{n}}{n h_{n}} \sum_{i=1}^{n} \frac{1}{Y_{i}} K\left(\frac{x-Y_{i}}{h_{n}}\right),
\end{equation}
Here  $h_n$ denotes a sequence of positive bandwidths such that $h_n \to 0$ and $nh_n \to \infty$ as $n \to \infty$, and $K(\cdot)$ represents a kernel function. \\
Let  $Y_{n:1}\leq \ldots \leq Y_{n:n} $ be the order statistics corresponding to $Y_1, \dots, Y_n$, then the empirical estimator of $Q$ given by \cite{sen1984asymptotic} is,
\begin{equation}\label{sen}
    Q_n(u)=Y_{k_n:n},~~0 < u <1
\end{equation}
where $k_n$ is a random integer depending on all of the order statistics and is defined by $k_n= max \left\{ k:\sum_{i=1}^{k}Y_{i:n}^{-1} \leq u\sum_{i=1}^{n}Y_{i:n}^{-1}\right\}$. Whenever such a $k$ does not exist, they have chosen $k_n=1$ to eliminate the technical difficulty.\\
Following \cite{parzen1979nonparametric} and \cite{jones1992estimating} defined the quantile density function in terms of $f(Q(u))$. By differentiating \eqref{quantilefunctn}, we have
\begin{equation*}
    q(u)=\frac{1}{f(Q(u))},~~ 0<u<1
\end{equation*}
where $q(u)= Q^{\prime}(u)$ represents the quantile density function, which plays a significant role in statistical modeling and analysis. \\
\cite{akbari2019nonparametric} proposed two estimators of the quantile density function under length-biased sampling. The first estimator was given as,
\begin{equation}\label{est1}
    q_{1n}(u)=\frac{1}{f_n(Q_n(u))},~~ 0<u<1
\end{equation}
where $f_n(x)$ denotes the kernel density estimate of $f(x)$, and $Q_n(u)$ represents the empirical estimate of the quantile function $Q(u)$.\\
Another estimator of quantile density function $q(\cdot)$ using the kernel method is given by \cite{akbari2019nonparametric} is,
\begin{equation}
\begin{aligned}\label{est2}
    q_{2n}(u)&= \frac{1}{h_n}\int_{0}^{1}K\left(\frac{p-u}{h_n}\right)dQ_n(p)\\
    &=\frac{1}{h_n}\sum_{i=1}^{n-1}K\left(\frac{T_i-u}{h_n}\right)\left(Y_{i+1:n}-Y_{i:n}\right)
\end{aligned}
\end{equation}
where the empirical estimator of $Q_n$ is a random step function that jumps up by $\left(Y_{i+1:n}-Y_{i:n}\right)$ at each random variables $T_i$, and $T_i= \frac{\sum_{j=1}^{i}Y_{j:n}^{-1}}{\sum_{i=1}^{n}Y_{i:n}^{-1}},(i=1, \ldots, n-1).$ The estimator defined in \eqref{est1} and \eqref{est2} is a strong uniformly consistent estimator of $q(u)$.

\section*{Assumptions}
\begin{enumerate}
    \item Let $K$ be a kernel function of bounded variation, vanishing outside of the interval $\left(-1,1\right)$, and
    \begin{equation*}
        \int_{-1}^{1}K(x)\;dx=1,~~~\int_{-1}^{1}xK(x)\;dx=0,~~~\int_{-1}^{1}x^2K(x)\;dx<\infty
    \end{equation*}
    \item Density function $f$ is differentiable in a neighborhood around $Q(u)$, and the $fQ(u)$ is finite and strictly positive for all $0<u<1$, and $E(Y^{-2})=\int_{0}^{\infty}y^{-2}\;dG(y)$ is finite.
    \item $f$ is twice differentiable, and its second derivative $f''$ is continuous and bounded.
    \item The quantile density function $q$ has a second derivative $q''$ that is bounded.
    \item $\frac{h_n\log{h_n}^{-1}}{\log \log{n}}={o}(1)$
    \item $h_n^{-1}\sqrt{\frac{\log \log n}{n}}={o}(1)$
\end{enumerate}
\section*{Remark}
\noindent
The first assumption ensures that the standard conditions required for the kernel method are satisfied, which is the case for a wide variety of kernels. Assumptions (2) and (3) follow the conditions established by \cite{sen1984asymptotic} to prove the uniform strong consistency of the estimator $Q_n$. Lastly, assumption (4) will be useful in our subsequent analysis. The conditions in the bandwidth $h_n$ in assumptions (5) and (6) are not restrictive. Considering $h_n \sim n^{-\alpha}$, if we choose $0<\alpha<1/2$, then both of the requirements are satisfied.

 \section{Entropy estimation}\label{sec3}
 This section proposes two estimators for the quantile-based Shannon entropy function in the length-biased context defined in \eqref{qse} based on the estimators mentioned in Section \ref{sec2}. The primary objective of this section is to demonstrate the strong uniform consistency of the proposed estimators. \\Let $\left(Y_{1}, Y_{2}, \ldots, Y_{n}\right)$  that are independent and identically distributed from the length-biased random variable $Y$.
Then a plug-in estimator for the quantile-based Shannon entropy based on \cite{akbari2019nonparametric} as defined in \eqref{est1} is given by,
\begin{equation} \label{myest1}
  \begin{aligned}
   \hat{\xi}_1(u)&=\int_{0}^{1} \log \hat{q}(u) \, du
   &=\int_{0}^{1} \log \frac{1}{f_nQ_n(u)} \,du
\end{aligned}
  \end{equation}
  
\begin{theorem}
    Suppose the assumptions hold (1)-(3) hold. Let $h_n$ be a sequence of positive bandwidths satisfying the assumptions (4) and (5). Then,
   \begin{equation*}
        \sup_x \left| \hat{\xi}_1(u)-\xi(u) \right | = \mathcal{O}(a_n\gamma_n^{-1})~~~~~a.s.
 \end{equation*}    
    Where $a_n= h_n^2 \vee h_n^{-1 }n^{-1/2}(\log \log n)^{1/2}$.
  \end{theorem}

\begin{proof}
    We have,
    \begin{equation*}
        \hat{\xi}_1(u)=\int_{0}^{1} \log \frac{1}{f_nQ_n(u)}\,du
    \end{equation*}
Equivalently,
\begin{equation*}
\begin{aligned}
    \hat{\xi}_1(u)-\xi(u)&= -\int_{0}^{1} \log {f_nQ_n(u)}\,du + \int_{0}^{1} \log {fQ(u)}\,du\\
    &=\int_{0}^{1} \log {fQ(u)}\,du -\int_{0}^{1} \log {f_nQ_n(u)}\,du
\end{aligned}
    \end{equation*}
    We have,
    \begin{equation}\label{eq9}
        \left| \hat{\xi}_1(u)-\xi(u) \right | = \int_{0}^{1} \left | \log \frac{fQ(u)}{f_nQ_n(u)} \right | \,du
    \end{equation}
    Since for all $z \geq 0$, $|\log z| \leq | \frac{1}{z} -1 |$, we get,
    \begin{equation*}
         \left | \log \frac{fQ(u)}{f_nQ_n(u)} \right | \leq \left | \frac{1}{\frac{fQ(u)}{f_nQ_n(u)}} -1 \right | \leq \left | \frac{f_nQ_n(u)-fQ(u)}{fQ(u)}  \right |
    \end{equation*}
    For any $X \in A_n$, one can see that ,
    \begin{equation*}
        fQ(u) \geq \gamma_n
       \end{equation*}
       We can easily obtain from these relation that,
    \begin{equation}\label{eq10}
        \sup_{x \in A_n} \left | \log \frac{fQ(u)}{f_nQ_n(u)} \right | \leq  \frac {\sup_x |f_nQ_n(u)-fQ(u)|} {\gamma_n}
        \end{equation}
        Substituting equation \eqref{eq10} in \eqref{eq9}, we get
        \begin{equation*}
            \sup_x \left| \hat{\xi}_1(u)-\xi(u) \right | \leq \frac {\sup_x |f_nQ_n(u)-fQ(u)|} {\gamma_n} \int_{0}^{1} \,du
        \end{equation*}
        Using Lemma A.2 of \cite{akbari2019nonparametric} we get,
    \begin{equation*}
        \sup_x \left| \hat{\xi}_1(u)-\xi(u) \right | = \mathcal{O}(a_n\gamma_n^{-1})~~~~~a.s.
 \end{equation*}    
    Where $a_n= h_n^2 \vee h_n^{-1 }n^{-1/2}(\log \log n)^{1/2}$.
\end{proof}

Next, we consider a plug-in estimator for quantile-based entropy function as defined in \eqref{est2}, given by,
\begin{equation} \label{myest2}
     \hat{\xi}_2(u)=\int_{0}^{1} \log \left( \frac{1}{h_n}\int_{0}^{1}K\left(\frac{p-u}{h_n}\right)\, dQ_n(p)\right) \,du
\end{equation}

The strong uniform consistency of the above estimator \eqref{myest2} is given in Theorem 3.2

\begin{theorem}
    Suppose the assumptions (1), (3) and (4) holds. Then as $n \to \infty$, it can be obtained that, for any $\epsilon>0$,
   \begin{equation*}
    \sup_{\epsilon \leq u \leq 1-\epsilon} \left| \hat{\xi}_2(u)-\xi(u)\right|=\mathcal{O}\left(h_n^{-1}\sqrt{\frac{\log \log n}{n}}\vee h_n^2\right)~~~~~a.s.
\end{equation*}
\end{theorem}
\begin{proof}
    We have,
    \begin{equation*}
        \hat{\xi}_2(u)= \int_{0}^{1} \log \left( \frac{1}{h_n}\int_{0}^{1}K\left(\frac{p-u}{h_n}\right)\, dQ_n(p)\right) \,du
    \end{equation*}
Now,
\begin{equation*}
\begin{aligned}
     \hat{\xi}_2(u)-\xi(u)&=\int_{0}^{1} \log\left(\frac{q_{2n}(u)}{q(u)}\right) \, du \\
     &=\int_{0}^{1} \log \left( \frac{\frac{1}{h_n}\int_{0}^{1}K \left(\frac{p-u}{h_n}\right)q_n(p)dp}{q(u)}\right) \, du \\
     &=\int_{0}^{1} \log \left( \frac{\frac{1}{h_n}\int_{0}^{1}\frac{K \left(\frac{p-u}{h_n}\right)}{f_nQ_n(p)}dp}{q(u)}\right) \, du
     \end{aligned}
\end{equation*}
We have,
\begin{equation*}
\begin{aligned}
    \Tilde{q}(u)-q(u)&=\frac{1}{h_n}\int_{0}^{1}\frac{K \left(\frac{p-u}{h_n}\right)}{f_nQ_n(p)}dp-q(u)\\
    \Tilde{q}(u)-q(u)&=\frac{1}{h_n}\int_{0}^{1}\frac{K \left(\frac{p-u}{h_n}\right)}{f_nQ_n(p)}dp-\frac{1}{h_n}\int_{0}^{1}\frac{K \left(\frac{p-u}{h_n}\right)}{fQ(p)}dp + \frac{1}{h_n}\int_{0}^{1}\frac{K \left(\frac{p-u}{h_n}\right)}{fQ(p)}dp-\frac{1}{fQ(u)}\\
    &=\frac{1}{h_n}\int_{0}^{1}K \left(\frac{p-u}{h_n}\right)\, dQ_n(p)-\frac{1}{h_n}\int_{0}^{1}K \left(\frac{p-u}{h_n}\right)\, dQ(p)\\ &+\frac{1}{h_n}\int_{0}^{1}K \left(\frac{p-u}{h_n}\right)\, dQ(p)-q(u)\\
    &= I+II
    \end{aligned}
\end{equation*}
Consider,
\begin{equation*}
    \begin{aligned}
        I&=\frac{1}{h_n}\int_{0}^{1}K \left(\frac{p-u}{h_n}\right)\, d(Q_n(p)-Q(p))\\
         \end{aligned}
\end{equation*}

\text{Applying integration by parts and Remark (3) of \cite{akbari2019nonparametric} we get,}
\begin{equation*}
\begin{aligned}
I&\leq \frac{1}{h_n}\int_{0}^{1}|Q_n(p)-Q(p)|dK\left(\frac{p-u}{h_n}\right)\\
&\leq \frac{1}{h_n} \sup_{0<p<1} |Q_n(p)-Q(p)|\\
&= \mathcal{O}\left(h_n^{-1}\sqrt{\frac{\log \log n}{n}}\right)~~ a.s.
 \end{aligned}
\end{equation*}
To deal with II, we have,
\begin{equation*}
\begin{aligned}
  II&=\frac{1}{h_n}\int_{0}^{1}K \left(\frac{p-u}{h_n}\right)\, dQ(p)-q(u)\\
    \end{aligned}
\end{equation*}
Putting $\frac{p-u}{h_n}= \nu$, we get,
\begin{equation*}
    II=\int_{-\frac{u}{h_n}}^{\frac{1-u}{h_n}}K(\nu)q(u+\nu{h_n})\,d\nu -q(u)
\end{equation*}
Now using the Taylor series expansion $q(u+\nu{h_n})$ around $u$ we can write $q(u+\nu{h_n})=q(u)+\nu{h_n}q^{'}(u)+ \ldots$, Since $K$ satisfies Assumption 1, being of bounded variation and $q$ satisfies Assumption 4 with a bounded second derivative $q^{''}$, we get $II=\mathcal{O}(h_n^2)$ uniformly in $\epsilon \leq u \leq 1-\epsilon$.
Accordingly, this results,
\begin{equation*}
    \sup_{\epsilon \leq u \leq 1-\epsilon} \left| \Tilde{q}(u)-q(u)\right|=I+II=\mathcal{O}\left(h_n^{-1}\sqrt{\frac{\log \log n}{n}}\vee h_n^2\right)~~~~a.s.
\end{equation*}
Hence the required result easily follows.
\end{proof}

 \section{Numerical analysis}\label{sec4}
 \subsection{Simulation}\label{simulation}
This section presents numerical simulations to evaluate the performance of the two estimators introduced in Section \ref{sec2}. For this purpose, we consider two families of distributions for generating length-biased data.
First, random samples of various sizes 
$n$ are drawn from the Govindarajulu distribution with quantile function,
\begin{equation}\label{GR}
   Q(u)= \theta + \sigma((\beta+1)u^{\beta}-\beta{u^{\beta+1}}),~~ \theta,\sigma,\beta>0,~~0 \leq u \leq 1.
\end{equation}
This distribution is a well-known quantile-based model commonly used to describe lifetime data with a bathtub-shaped hazard function. The simulation study is conducted using the following parameter settings: 
\begin{equation*}
\theta=0,~\sigma=1,~\beta=1/4 
\end{equation*}
\begin{equation*}
     \theta=0,~\sigma=3/4,~\beta=1/4
   \end{equation*}
\begin{equation*}
     \theta=0,~\sigma=1,~\beta=1
\end{equation*}
In this simulation study, the estimators are computed using the Epanechnikov function as the kernel function and the bandwidth parameter $h_n$ is selected according to the rule-of-thumb method proposed by \cite{borrajo2017bandwidth} for length-biased data.
We generate 500 simulated samples from Govindarajulu distributions for varying values of $n=50,100,300,400$.

\begin{table}[t]
    \centering
    \caption{Abs. Bias and MSE for $\hat{\xi_1}$  for Govindarajulu distribution.}
    \renewcommand{\arraystretch}{1.2}
\setlength{\tabcolsep}{12pt}

    \begin{tabular}{|c|c|cc|}
    \hline
        \textbf{Govindarajulu model } & \textbf{n} & \textbf{MSE} & \textbf{Abs. Bias }   \\ \hline
        (0,1,1/4) & 50 & 0.0833 & 0.2043  \\ 
         & 100 & 0.0612 & 0.2031  \\ 
         & 300 & 0.0482 & 0.1934  \\ 
         & 400 & 0.0457 & 0.1892  \\ 
         &  &  &   \\ 
        (0,3/4,1/4) & 50 & 0.0845 & 0.2072  \\ 
         & 100 & 0.0624 & 0.2060  \\ 
         & 300 & 0.0494 & 0.1963  \\ 
         & 400 & 0.0468 & 0.1921  \\ 
         &  &  &   \\ 
        (0,1,1) & 50 & 0.0516 & 0.0100  \\ 
         & 100 & 0.0321 & 0.0052  \\ 
         & 300 & 0.0203 & 0.0022  \\ 
         & 400 & 0.0141 & 0.0011  \\ 
           \hline
    \end{tabular}
    \label{table 1}
\end{table}

\begin{table}[t]
    \centering
    \caption{Abs. Bias and MSE for $\hat{\xi_2}$ for Govindarajulu distribution.}
    \renewcommand{\arraystretch}{1.2}
\setlength{\tabcolsep}{12pt}

    \begin{tabular}{|c|c|cc|}
    \hline
        \textbf{Govindarajulu model } & \textbf{n} & \textbf{MSE} & \textbf{Abs. Bias } \\ \hline
        (0,1,1/4) & 50 & 0.0834 & 0.1165  \\ 
        & 100 & 0.0358 & 0.0514  \\ 
         & 300 & 0.0092 & 0.0063  \\ 
         & 400 & 0.0072 & 0.0061  \\ 
         &  &  &   \\ 
        (0,3/4,1/4) & 50 & 0.0960 & 0.1528  \\ 
        & 100 & 0.0388 & 0.0685  \\ 
        & 300 & 0.0171 & 0.0157  \\ 
        & 400 & 0.0073 & 0.0063  \\ 
        &  & &   \\ 
        (0,1,1) & 50 & 0.0642 & 0.1845  \\ 
         & 100 & 0.0270 & 0.1058  \\ 
         & 300 & 0.0066 & 0.0315  \\ 
         & 400 & 0.0048 & 0.0207  \\ 
        \hline
    \end{tabular}
    \label{table 2}
\end{table}

From Tables~\ref{table 1} and~\ref{table 2}, it is evident that as the sample size 
$n$ increases, both the Mean Squared Error (MSE) and the Absolute Bias (Abs. Bias) decrease. Among the two estimators, $\hat{\xi}_2(u)$  consistently shows lower MSE and Absolute Bias compared to $\hat{\xi}_1(u)$
, indicating better performance.

To further evaluate the estimators, we examine another quantile-based model without a closed-form cumulative distribution function: the Generalized Lambda Distribution (GLD), defined by parameters $(\lambda_1,\lambda_2,\lambda_3,\lambda_4)$  , with the quantile function given by:
\begin{equation}\label{Qgld}
Q(u) = \lambda_1 + \frac{u^{\lambda_3} - (1 - u)^{\lambda_4}}{\lambda_2}, \quad 0 < u < 1
\end{equation}
In \eqref{Qgld}, $\lambda_1$ is a location parameter, $\lambda_2$ is a scale parameter, while $\lambda_3$ and $\lambda_4$ determine the shape of the density function.
The simulation study is conducted using the following parameter configurations:
\begin{equation*}
\lambda_1 = 2,\ \lambda_2 = 1,\ \lambda_3 = 2,\ \lambda_4 = 6
\end{equation*}
\begin{equation*}
\lambda_1 = 2,\ \lambda_2 = 1,\ \lambda_3 = 3,\ \lambda_4 = 5
\end{equation*}
\begin{equation*}
\lambda_1 = 3,\ \lambda_2 = 2,\ \lambda_3 = 1,\ \lambda_4 = 5
\end{equation*}

For each of these models, the MSE and Absolute Bias of the estimators  $\hat{\xi}_1(u)$
and $\hat{\xi}_2(u)$ were computed based on 500 simulated samples for various values of 
$n$. The results are summarized in Tables~\ref{table 3} and~\ref{table 4}.

\begin{table}[t]
    \centering
    \caption{Abs. Bias and MSE for $\hat{\xi_1}$  for Generalised lambda distribution.}
    \renewcommand{\arraystretch}{1.2}
\setlength{\tabcolsep}{12pt}

    \begin{tabular}{|c|c|cc|}
    \hline
         \textbf{GLD model } & \textbf{n} & \textbf{MSE} & \textbf{Abs. Bias }   \\ \hline
        (2,1,2,6) & 50 & 0.0213 & 0.0975  \\ 
         & 100 & 0.0109 & 0.0757  \\ 
         & 200 & 0.0058 & 0.0551  \\ 
         & 400 & 0.0047 & 0.0540  \\ 
         & 500 & 0.0045 & 0.0519  \\ 
         &  &  &   \\ 
        (2,1,3,5) & 50 & 0.0165 & 0.0743  \\ 
         & 100 & 0.0073 & 0.0498  \\ 
         & 200 & 0.0034 & 0.0267  \\ 
         & 400 & 0.0026 & 0.0260  \\ 
         & 500 & 0.0025 & 0.0229  \\ 
         &  &  &   \\ 
        (3,2,1,5) & 50 & 0.0141 & 0.0393  \\ 
         & 100 & 0.0060 & 0.0233  \\ 
         & 200 & 0.0035 & 0.0086  \\ 
         & 400 & 0.0033 & 0.0081  \\ 
         & 500 & 0.0028 & 0.0054  \\ 
          \hline
    
    \end{tabular}
    \label{table 3}
\end{table}

\begin{table}[t]
    \centering
    \caption{Abs. Bias and MSE for $\hat{\xi_2}$  for Generalised lambda distribution.}
    \renewcommand{\arraystretch}{1.2}
\setlength{\tabcolsep}{12pt}

    \begin{tabular}{|c|c|cc|}
    \hline
       \textbf{GLD model } & \textbf{n} & \textbf{MSE} & \textbf{Abs. Bias } \\ \hline
        (2,1,2,6) & 50 & 0.0377 & 0.1751  \\ 
         & 100 & 0.0190 & 0.1276  \\ 
         & 200 & 0.0101 & 0.0943  \\ 
         & 400 & 0.0064 & 0.0765  \\ 
         & 500 & 0.0058 & 0.0732  \\ 
         &  &  &   \\ 
        (2,1,3,5) & 50 & 0.0351 & 0.1698  \\ 
         & 100 & 0.0177 & 0.1232  \\ 
         & 200 & 0.0094 & 0.0917  \\ 
         & 400 & 0.0061 & 0.0749  \\ 
         & 500 & 0.0055 & 0.0718  \\ 
         &  &  &   \\ 
        (3,2,1,5) & 50 & 0.0249 & 0.1243  \\ 
         & 100 & 0.0110 & 0.0831  \\ 
         & 200 & 0.0047 & 0.0542  \\ 
         & 400 & 0.0023 & 0.0401  \\ 
         & 500 & 0.0020 & 0.0374  \\ \hline
    \end{tabular}
    \label{table 4}
\end{table}

The results indicate that $\hat{\xi}_1(u)$ and $\hat{\xi}_2(u)$ have possessed reasonable performance.
Based on the values presented in these tables, we can draw the following conclusions:
\begin{itemize}
    \item As expected, the MSE and Abs. Bias value of the two estimators decreases as the sample size $n$ increases.
    \item In the case of the Govindarajulu model, the estimator $\hat{\xi_2}$ perform better in terms of MSE.
    
\end{itemize}

 \subsection{Comparitive study}\label{compstudy}
This section compares our estimators with existing estimators from the literature. Considerable research has been devoted to the nonparametric estimation of Shannon entropy. Here, we focus on two kernel-based entropy estimators for length-biased data, as introduced by \cite{rajesh2022kernel}, and carry out a comprehensive study with the proposed estimators using quantile models that lack a closed-form distribution function, as discussed in Section \ref{simulation}.
\par
 Assuming that $(Y_1,Y_2, \ldots,Y_n)$ a sample of independent and identically distributed observations from length-biased random variable $Y$ with common distribution function $G$. Then the kernel density estimator of $f$ due to \cite{franklin1988comparioson}, is given by,
\begin{equation}\label{bhattest}
    \hat{f}_1(x)=\frac{\frac{1}{h}\sum_{i=1}^{n}K\left(\frac{x-Y_i}{h}\right)}{x\sum_{i=1}^{n} Y_i^{-1}}
\end{equation}
 where $K(\cdot)$ is a kernel function and $h$ is the bandwidth. \cite{jones1991kernel} proposed an improved estimator,
 \begin{equation}\label{jonesest}
     \hat{f}_2(x)=\frac{\frac{1}{h}\sum_{i=1}^{n}K\left(\frac{x-Y_i}{h}\right)}{\sum_{i=1}^{n} Y_i^{-1}}
 \end{equation}
Then two integral estimators of \cite{rajesh2022kernel} for $H(X)$ are given by, 
 \begin{equation}\label{entropyest}
     \hat{H}_l= -\int_{0}^{\infty} (\log\hat{f}_l(x))\hat{f}_l(x),~~ l=1,2,
 \end{equation}
where $\hat{f}_l$ is the estimator given in \eqref{bhattest} and \eqref{jonesest} respectively. For comparison, length-biased random samples are generated from the Govindarajulu distribution with parameters $\theta=0,~\sigma=3/4,~\beta=1/4$ and Generalized lambda distribution with parameters $\lambda_1=2,\lambda_2=1,\lambda_3=3,\lambda_4=5$. The estimated results are presented in Table \ref{tab:mse_results}.
\par
As anticipated, Table \ref{tab:mse_results} shows that for the Generalized lambda model with parameters $(2,1,3,5)$, the estimator $\hat{\xi_1}$ outperforms the other estimators in terms od Mean Squared Error. In contrast, for the Govindarajulu model with parameters
 $(0,3/,4,1/4)$, the estimator $\hat{\xi_2}$ has the least MSE as compared to the other estimators $\hat{\xi_1}$, $\hat{H}_1$ and $\hat{H}_2$.

\begin{table}[t]
\centering
\caption{Mean Squared Error (MSE) comparison for estimators defined in \eqref{entropyest} for different quantile models and sample sizes.}

\renewcommand{\arraystretch}{1.2}
\setlength{\tabcolsep}{12pt}

\begin{tabular}{|l|c|cccc|}
\hline
Model                                    & $n$   & \multicolumn{4}{c|}{MSE}           \\ \hline
                                         &     & $\hat{H}_1$     & $\hat{H}_2$     & $\hat{\xi}_{1}$   & $\hat{\xi}_{2}$    \\
Govindarajulu(0,3/4,1/,4)                & 50  & 0.1023 & 0.1026 & 0.0845 & 0.0960 \\
                                         & 100 & 0.0659 & 0.0660 & 0.0624 & 0.0388 \\
                                         & 300 & 0.0399 & 0.0400 & 0.0494 & 0.0171 \\
                                         & 400 & 0.0352 & 0.0352 & 0.0468 & 0.0073 \\
                                         &      &       &         &       &        \\
GLD (2,1,3,5) & 50  & 0.0225 & 0.0206 & 0.0165 & 0.0351 \\
                    & 100 & 0.0140 & 0.0126 & 0.0073 & 0.0177 \\
                     & 400 & 0.0055 & 0.0051 & 0.0026 & 0.0061 \\
                    & 500 & 0.0045 & 0.0042 & 0.0025 & 0.0055\\
\hline
\end{tabular}
\label{tab:mse_results}
\end{table}

 \subsection{Real data analysis}\label{realdata}
In this analysis, we examine length-biased data comprising width measurements of 46 shrubs, obtained through line-transect sampling as detailed in Table 3 of \cite{muttlak1990ranked}. This sampling method assigns a higher inclusion probability to shrubs with greater widths, as the probability of selection is directly proportional to a shrub's width. Consequently, the dataset exemplifies length-biased sampling, where larger units are more likely to be included due to the inherent nature of the sampling mechanism. We have fitted a Power-Pareto distribution with shape parameters $\lambda_1$ and $\lambda_2$ and scale parameter is $C$.\par

\begin{figure}[ht!]
    \centering
\includegraphics[width= 15cm]{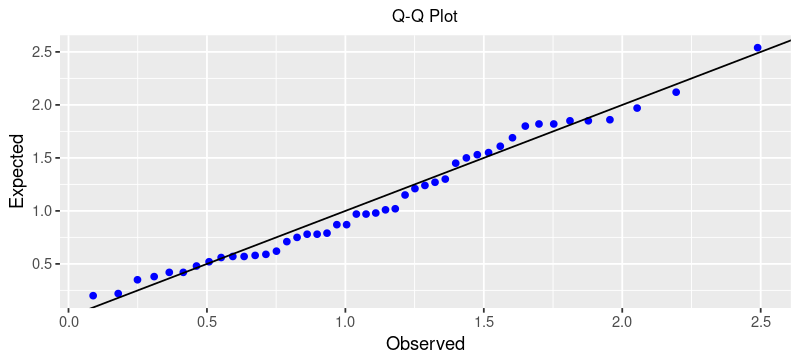} 
\caption{\textit{Q-Q} plot for shrubs data.}
\label{fig1}
\end{figure}  

The goodness of fit was assessed both graphically, using the $Q-Q$ plot (see Figure \ref{fig1}), and theoretically using the
 Kolmogorov-Smirnov (KS) test. The calculated value of the KS statistic is $0.1099$, which is less than the table value of $0.2292$. In both cases, the results indicate that the model is not rejected, suggesting that the Power-Pareto distribution provides a better fit to the data.  The corresponding parameter estimates obtained through the maximum likelihood method are as follows:
 
 \begin{equation*}
     \hat{C}= 1.5827,~~\hat{\lambda}_1=0.6368,~~ \hat{\lambda}_2=0.1016
 \end{equation*}

 The true value of the corresponding quantile-based entropy, as defined in \eqref{qse}, is $0.7570$. The nonparametric estimate of the quantile-based entropy based on the data is $\hat{\xi}_1 = 0.7541$, and for $\hat{\xi}_2$, it is $0.3677$.It is evident that the estimate obtained using $\hat{\xi}_1$ is closer to the true value compared to that obtained using $\hat{\xi}_2$.

 \section{Conclusion}\label{conclusion}
This work presents two nonparametric methods for estimating quantile-based Shannon entropy in the context of length-biased sampling, offering practical applications in the analysis and modeling of lifetime data. Using $q_{1n}(u)$ and $q_{2n}(u)$ from \cite{akbari2019nonparametric} as plug-in estimators, we proposed two new estimators, $\hat{\xi}_1(u)$ and $\hat{\xi}_2(u)$, and established their strong uniform convergence. To assess their performance, we conducted simulations with length-biased samples from various quantile distributions, including the Generalized lambda and Govindarajulu distributions. A comparison of mean square error (MSE) with the estimator from \cite{rajesh2022kernel} demonstrated that our proposed estimators achieve lower MSE, highlighting their improved accuracy. Additionally, their effectiveness was validated through real data analysis, confirming their practical applicability.

\section*{Funding details}
This work was supported by the Department of Science and Technology, Government of India, under the "WISE Fellowship for PhD" program, with file number DST/WISE-PhD/PM/2023/60(G).

\section*{Acknowledgements}
The first author sincerely acknowledges the financial support from the Department of Science and Technology, Government of India, under the "WISE Fellowship for PhD" program, which supported this research.

\section*{Declaration of competing interest}
The authors declare that they have no known competing financial interests or personal relationships that could have
appeared to influence the work reported in this paper.

\bibliographystyle{elsarticle-harv} 
 \bibliography{qll}

\end{document}